\documentclass[11pt,reqno]{amsart}
\usepackage{amsmath,amsfonts,amscd,amssymb,epsf}
\newtheorem{theorem}{Theorem}[section]
\newtheorem{proposition}[theorem]{Proposition}
\newtheorem{lemma}[theorem]{Lemma}

\newtheorem{remark}[theorem]{Remark}
\newtheorem{definition}[theorem]{Definition}

\numberwithin{equation}{section}
\renewcommand{\L}{\mathcal{L}}
\newcommand\Zop{\mathbb{Z^{\mathrm{odd}}_+}}

\renewcommand{\d}{\operatorname{d}}

\renewcommand{\d}{\partial}

\newcommand{\Z}{{\mathbb{Z}}}
\newcommand{\N}{{\mathbb{N}}}

\renewcommand{\L}{\mathcal{L}}
\newcommand{\M}{\mathcal{M}}

\newcommand{\B}{\mathcal{B}}

\newcommand{\pa}{\partial}

\def\u{1\kern -0.7em \hbox {1}}

\def\Z{\bf Z}

\begin{document}
\title {Dispersionless and multicomponent BKP hierarchies with quantum torus symmetries}
\author{
Chuanzhong Li\  } \dedicatory {  Department of Mathematics,
 Ningbo University, Ningbo 315211, China,\\
Email:lichuanzhong@nbu.edu.cn}

\begin{abstract}
In this article, we will construct  the additional perturbative quantum torus symmetry of the dispersionless BKP hierarchy basing on the $W_{\infty}$ infinite dimensional  Lie symmetry.
 These results show that the complete quantum torus symmetry is broken from the BKP hierarchy to its dispersionless hierarchy. Further  a series of  additional flows of the multicomponent BKP hierarchy will be defined and these flows constitute an $N$-folds direct product of the
positive half of the quantum torus symmetries.
\end{abstract}

\maketitle
\tableofcontents

\section{Introduction}

The KP hierarchy  is one of the most important integrable hierarchies\cite{DKJM} and it arises in many different fields of mathematics and physics such as the enumerative algebraic geometry, topological field and string theory.   One of the most important study on the KP hierarchy is the  theoretical  description of the solutions  of the KP hierarchy using Lie groups and Lie algebras such as in \cite{DKJM,dickey,kac}, which is closely related to the infinite dimensional Grassmann manifolds \cite{sato1,sato2}.

In \cite{2},   Date,   Jimbo,   Kashiwara and   Miwa extended their work on the KP hierarchy to the multicomponent KP hierarchy.
In \cite{1}, Takasaki and Takebe derived   a series of differential Fay identities for the multicomponent KP hierarchy from the bilinear identities  and they showed that their dispersionless limits give rise to the universal Whitham hierarchy. Besides the multicomponent KP hierarchy, the extended and reduced multicomponent Toda hierarchies attract a lot of studies \cite{manas,EMTH,EZTH}.

Additional symmetries have been studied in the explicit form of the
additional flows of the KP hierarchy  by Orlov and Shulman
\cite{os1}. This kind of additional  flows depend on dynamical variables
explicitly and constitute a centerless
$W_{1+\infty}$ algebra which is closely related to the theory of matrix
models\cite{D witten,Douglas} by the Virasoro constraint and
string equations.
 As a generalization of the Virasoro algebra, the Block algebra was studied a lot in the field of Lie algebras and it was studied intensively in
references\cite{Block}-\cite{Su}. In another paper\cite{ourBlock}, we give a novel Block type
additional symmetry of the bigraded Toda hierarchy(BTH). Later we did a series of studies on integrable systems and Block algebras such as in \cite{dispBTH,blockDS,blockdDS}. After the quantization, the Block Lie algebra becomes the so-called quantum torus Lie algebra which can be seen in several recent  references as \cite{takasakiquantum,torus}.

It is well known that the KP hierarchy has two sub-hierarchies, i.e.
the BKP hierarchy and CKP hierarchy. About the BKP hierarchy, a lot of studies on additional symmetries have been done, such as additional symmetries of the BKP hierarchy\cite{TU}, dispersionless BKP hierarchy\cite{disBKP,TUdbkp}, two-component BKP hierarchy\cite{LWZ} and its reduced hierarchies \cite{blockDS,blockdDS}, supersymmetric BKP hierarchy \cite{NPB} and so on.
Dispersionless integrable systems\cite{Takasaki} are very important in
the study of all kinds of nonlinear sciences in
physics, particularly in the application on the
topological field theory \cite{Takasakicmp} and matrix models
\cite{Kodama-Pierce,matrix model}. In particular,
dispersionless integrable systems have many typical properties  such as the Lax pair, infinite conservation
laws, symmetries and so on.  On dispersionless integrable systems, we also did several studies such as \cite{dispBTH,blockdDS}.
With the above preparation, we should pay our attention to the quantum torus type additional symmetry of the multi-component BKP hierarchy\cite{multiBKP} and dispersionless BKP hierarchy\cite{disBKP,disBKP2} from the points of the multi-component generalization of Lie algebras and the importance of dispersionless integrable systems.

 In the next section, we firstly review the Lax equations of the BKP and dispersionless BKP hierarchies. In Section 3, under the basic Sato theory, we construct  the additional symmetries of the dispersionless BKP hierarchy and the symmetries form a  perturbative quantum torus    Lie algebra.
 In Section 4, we recall the Lax equation of the multicomponent BKP hierarchy. In Section 5, we construct  the additional symmetry of the multicomponent BKP hierarchy which turns out to be in an $N$-folds direct product of the infinite dimensional complete  quantum torus  Lie algebra.

\section{ BKP hierarchy and dispersionless BKP hierarchy}

Similarly to the general way in describing the classical the BKP hierarchy
\cite{DKJM,dickey}, we will give a brief introduction of the BKP hierarchy.
We denote ``$*$" as a formal adjoint operation defined by $A^{*}=\sum(-1)^{i}\partial^{i}a_i$ for an arbitrary
scalar-valued pseudo-differential operator $A=\sum a_i\partial^{i}$, and $(AB)^{*}=B^{*}A^{*}$ for two scalar operators $A,B$.
Basing on the definition, the Lax operator of the BKP hierarchy is as
\begin{equation} \label{PhP}
L_B= \d+\sum_{i\ge1}v_i  \d^{-i},
\end{equation}
 such that

\begin{equation}\label{Bcondition}
L_B^*=-\d L_B\d^{-1}.
\end{equation}
The eq.\eqref{Bcondition} will be called the B type condition of the BKP hierarchy.

The  BKP hierarchy is defined by the following
Lax equations:
\begin{align}\label{bkpLax}
& \frac{\d  L_B}{\d t_k}=[(L_B^k)_+,  L_B],   \ k\in\Zop.
\end{align}
The ``+" in \eqref{bkpLax} means the nonnegative projection about the operator ``$\partial$" and ``-"  means the negative projection.
Note that the $\d/\d t_1$ flow is equivalent to the $\d/\d x$ flow, therefore it is reasonable to
assume $t_1=x$ in the next sections. The Lax operator $L_B$ can be generated by a dressing operator
$\Phi_B=1+ \sum_{k=1}^{\infty}\bar\omega_k
\partial^{-k}$ in the following way
\begin{equation}
L_B=\Phi_B  \partial \Phi_B^{-1},
\end{equation}
where $\Phi_B$
 satisfies
\begin{equation}\label{phipsi}
\Phi_B^*= \d\Phi_B^{-1} \d^{-1}.
\end{equation}
The dressing operator $\Phi_B$ needs to satisfy the following Sato equations
\begin{equation}
\dfrac{\partial \Phi_B}{\partial t_n}=-(L_B^{n})_-\Phi_B, \quad n=1,3,5, \cdots.
\end{equation}

Introduce
firstly the Lax function of dispersionless BKP hierarchy \cite{disBKP,disBKP2} as following
\begin{equation}\label{Lax operator}
L=k+u_1k^{-1}+u_3k^{-3}+\dots +\dots
\end{equation}
where the coefficients $u_1, u_3,\ldots$ of the Lax function are same as eq.\eqref{Lax}.
  The variables $u_j$ are functions of the real variable $x$. The Lax function $L$ can be
written as
  \[L=e^{ad\varphi}(k ),\ \ \ ad\varphi(\psi) = \{\varphi, \psi\}=\frac{\partial  \varphi}{\partial k}\frac{\partial  \psi}{\partial x}-\frac{\partial \varphi}{\partial k}\frac{\partial   \psi}{\partial x}.\]

The dressing function has the following form \begin{eqnarray}
&& \varphi=\sum_{n=1}^\infty \varphi_{2n} k^{-2n+1}.
\label{dressP}\end{eqnarray}

  The
dressing function $\varphi$ is unique up to adding some Laurent series about the variable $k$ with  coefficients which do not depend on $x$.
The dispersionless BKP hierarchy can be defined as following.
\begin{definition} \label{deflax}
The dispersionless BKP hierarchy consists of flows given in the  Lax pair by
\begin{equation}
\frac{\pa L}{\pa t_{n}}=\{\B_{n}, L\}=\frac{\partial  \B_{n}}{\partial k}\frac{\partial L}{\partial x}-\frac{\partial L}{\partial k}\frac{\partial  \B_{n}}{\partial x},\ \ n\in \Zop,
\end{equation}
where the functions
$\B_{n}$ are defined by
$\B_{n}:=(L^n )_+,\ \ \ n\in \Zop.$

\end{definition}
The ``+" here means the nonnegative projection about the variable ``$k$"   and ``-"  means the negative projection.
The Lax equation of the dispersionless BKP hierarchy can lead to the following dispersionless Sato equations in the next proposition.
\begin{proposition}$L$ is the Lax function of the
dispersionless BKP if and only if  there exists a Laurent series $\varphi$
({\it dressing fucntion}) which satisfies the equations
\begin{eqnarray}\label{sato dis}
    \nabla_{t_{n},\varphi} \varphi& =& -(\B_{n})_-, \ \ n\in \Zop,\end{eqnarray}

where
$$
 \nabla_{t_{n},
\psi} \varphi = \sum_{m=0}^\infty \frac1{(m+1)!} (ad\psi)^m \left(
\frac{\d \varphi}{\d  t_{n}} \right).
$$
The Laurent function $\varphi$  is  unique up to a transformation $\varphi \mapsto
H(\varphi,\psi)$, with a constant Laurent series $\psi =
\sum_{n=1}^\infty \psi_{2n} k^{-2n+1}$ ($\psi_{2n}$: constant),
where $H(X,Y)$ is the Hausdorff series  defined by
$$
    \exp(ad H(\varphi,\psi)) = \exp(ad\varphi) \exp(ad\psi).
$$

\end{proposition}

The simplest nontrivial flow in the dispersionless BKP hierarchy is the (2 + 1)-dimensional dispersionless BKP
equation :
\[3u_t + 15u^2u_x -5uu_y - 5u_x\d^{-1}u_y -\frac53\d^{-1}u_{yy} = 0.\]

\section{Perturbative quantum torus symmetry of the dispersionless BKP hierarchy}
In \cite{disBKP,TUdbkp}, they construct additional symmetries of the dispersionless BKP hierarchy.
In this section, we shall construct a specific kind of additional symmetries of the dispersionless BKP hierarchy and identify its nice quantum torus algebraic structure.

To this end, firstly we define the following dispersionless function $\Gamma$ and the dispersionless Orlov-Shulman's  function $M$  as
 \begin{equation}
\Gamma=x+\sum_{i\in \Zop}it_i\lambda^{i-1},\ \ M= e^{ad \varphi}( \Gamma)=\sum_{n=0}^{\infty}
(2n + 1)t_{2n+1}L^{2n} +\sum_{n=0}^{\infty}
v_{2n+2}L^{-2n-2}.
\end{equation}
The dispersionless Lax function $L$ and the dispersionless Orlov-Shulman's  function $M$ satisfy the following canonical relation
\[\{L,M\}=1.\]

Then basing on a quantum parameter $q$, the additional flows for the time variables $t_{m,n},t_{m,n}^*$ are
defined respectively as follows
\begin{equation}
 \nabla_{t_{m,2n+1},\varphi} \varphi=-(M^mL^{2n+1})_-,\  \nabla_{t^*_{m,n},\varphi} \varphi=-(e^{mM}(q^{nL}-q^{-nL}))_-,\ m,n \in \N,
\end{equation}
or equivalently rewritten as
\begin{equation}
\dfrac{\partial L}{\partial t_{m,2n+1}}=-\{(M^mL^{2n+1})_-,L\}, \qquad
\dfrac{\partial M}{\partial t_{m,2n+1}}=-\{(M^mL^{2n+1})_-,M\},
\end{equation}

\begin{equation}
\dfrac{\partial L}{\partial t^*_{m,n}}=-\{(e^{mM}(q^{nL}-q^{-nL}))_-,L\}, \qquad
\dfrac{\partial M}{\partial t^*_{m,n}}=-\{(e^{mM}(q^{nL}-q^{-nL}))_-,M\}.
\end{equation}

 Further one can also derive
\begin{equation}\label{MLK}
\partial_{t^*_{l,k}}(e^{mM}q^{nL}\ )=\{-(e^{lM}(q^{kL}-q^{-kL}))_-,e^{mM}q^{nL}\}.
\end{equation}

One can find the functions' set $\{M^mL^n,\ m,n\geq 0, \ 1\leq \alpha\leq N\}$ has an isomorphism with the operators' set $\{z^{n}\partial_z^m,\ m,n\geq 0, \ 1\leq \beta\leq N\}$ as
\begin{equation}
M^mL^n \qquad \mapsto\qquad  z^{n}\partial_z^m,
\end{equation}
with the following commutator
\begin{equation}
[M^mL^n,M^kL^l]=C_{ab}^{(mn)(kl)}M^aL^b.
\end{equation}

One can find the functions' set $\{e^{mM}q^{nL},\ m,n\geq 0, \ 1\leq \alpha\leq N\}$ has an isomorphism with the operators' set $\{q^{nz}e^{m\partial_z},\ m,n\geq 0, \ 1\leq \beta\leq N\}$ as
\begin{equation}
e^{mM}q^{nL} \qquad \mapsto\qquad  q^{nz}e^{m\partial_z},
\end{equation}
with the following commutator
\begin{equation}
[q^{nz}e^{m\partial_z},q^{lz}e^{k\partial_z}]=(q^{ml}-q^{nk})q^{(n+l)z}e^{(m+k)\partial_z}.
\end{equation}

The additional flows  $\dfrac{\partial }{\partial
t_{m,n}}$ will be proved later to  commute with the flows $\dfrac{\partial
}{\partial t_{ k}}$, i.e. $[\dfrac{\partial }{\partial
t_{m,n}},\dfrac{\partial
}{\partial t_{ k}}]=0$,  but they do not
commute with each other. They can form a  $W_{\infty}$ infinite dimensional   Lie algebra.
 This further leads to the commutativity of the additional flows $\dfrac{\partial }{\partial
t_{m,n}^*}$  with the flows $\dfrac{\partial
}{\partial t_{k}}$ and the additional flows $\dfrac{\partial }{\partial
t_{m,n}^*}$  themselves constitute a perturbative quantum torus algebra which will be proved later.

\begin{proposition}
The additional flows  $\partial_{t_{l,k}}$ are  symmetries of the dispersionless BKP hierarchy, i.e. they commute with all $\partial_{t_{ n}}$ flows of the  dispersionless BKP hierarchy.
\end{proposition}
\begin{proof}
According the action of  $\partial_{t_{l,k}}$ and $\partial_{t_{ n}}$ on the
dressing function $L$, then
\begin{eqnarray*}
&&[\partial_{t_{m,l}},\partial_{t_{ n}}]L\\
 &=&\partial_{t_{m,l}}\partial_{t_{ n}}e^{ad \varphi}(k)-
 \partial_{t_{ n}}\partial_{t_{m,l}}e^{ad \varphi}(k)\\
&=& \partial_{t_{m,l}}\{\nabla_{t_{ n}, \varphi} \varphi,e^{ad
\varphi}(k)\}-
 \partial_{t_{ n}}\{\nabla_{t_{m,l}, \varphi} \varphi, e^{ad \varphi}(k)\}\\
&=& \partial_{t_{m,l}}\{-(\B_{n})_-,L\}-
 \partial_{t_{ n}}\{\left(M^m
\ L^l\right)_{+}, L\}\\
&=&
\{-\{\left(M^m\ L^l\right)_{+},\B_{n}\}_-,L\}+\{-(\B_{n})_-,\{\left(M^m
\ L^l\right)_{+},
L\}\}- \\
&&\{\{-(\B_{n})_-,M^m\ L^l\}_{+}, L\}- \{\left(M^m
\ L^l\right)_{+}, \{-(\B_{n})_-,L\}\}\\
&=&
\{-\{\left(M^m\ L^l\right)_{+},\B_{n}\}_-,L\}+\{L,\{(\B_{n})_-,\left(M^m
\ L^l\right)_{+}
\}\}\\
&&- \{\{-(\B_{n})_-,M^m\ L^l\}_{+}, L\}
\\
&=&\{L,\{(\B_{n})_-,\left(M^m\ L^l\right)_{+}
\}+\{\left(M^m\ L^l\right)_{+},\B_{n}\}_-+\\
&&\{-(\B_{n})_-,M^m\ L^l\}_{+}\}\\
&=&0.
\end{eqnarray*}
Therefore the proposition holds.
\end{proof}
With the help of this proposition, we can derive the following theorem.
\begin{theorem}
The additional flows $\partial_{t^*_{l,k}}$ are  symmetries of the dispersionless BKP hierarchy, i.e. they commute with all $\partial_{t_{n}}$ flows of the  dispersionless BKP hierarchy.
\end{theorem}
\begin{proof}
According to the action of  $\partial_{t^*_{l,k}}$ and $\partial_{t_{n}}$ on the
dispersionless dressing function $L$,  we can rewrite the quantum torus flow $\partial_{t^*_{l,k}}$¡¡in terms of a combination of $\partial_{t_{p,s}}$ flows
\begin{eqnarray*}
\partial_{t^*_{l,k}}L &=& \{-(\sum_{p,s=0}^{\infty}\frac{l^p(k\log q)^{2s+1}M^pL^{2s+1}}{p!(2s+1)!})_-,L\}\\
 &=&\sum_{p,s=0}^{\infty}\frac{l^p(k\log q)^{2s+1}}{p!(2s+1)!}\partial_{t_{p,2s+1}}L,
\end{eqnarray*}
which further leads to
\begin{eqnarray*}
[\partial_{t^*_{l,k}},\partial_{t_{n}}]L &=& [\sum_{p,s=0}^{\infty}\frac{l^p(k\log q)^{2s+1}}{p!(2s+1)!}\partial_{t_{p,2s+1}},\partial_{t_{n}}]L\\
&=& \sum_{p,s=0}^{\infty}\frac{l^p(k \log q)^{2s+1}}{p!(2s+1)!}[\partial_{t_{p,2s+1}},\partial_{t_{n}}]L\\
 &=&0.
\end{eqnarray*}
Therefore the theorem holds.
\end{proof}
Because
\begin{eqnarray*}
[z^s\partial^p,z^b\partial^a]=\sum_{\alpha\beta}C_{\alpha\beta}^{(ps)(ab)}z^{\beta}\partial^{\alpha},
\end{eqnarray*}
and
\begin{equation}\label{zformala}
[q^{nz}e^{m\partial_z},q^{lz}e^{k\partial_z}]=(q^{ml}-q^{nk})q^{(n+l)z}e^{(m+k)\partial_z},
\end{equation}
therefore we can derive the following identity

\begin{eqnarray*}&&\{\{e^{nM}(q^{mL}-q^{-mL}),e^{lM}(q^{kL}-q^{-kL})\}\\
&&=(q^{ml}-q^{nk})e^{(n+l)M}(q^{(m+k)L}-q^{-(m+k)L})-(q^{ml}-q^{-nk})e^{(n+l)M}(q^{(m-k)L}-q^{(-m+k)L})\\
&&+(q^{ml}-q^{nk}+q^{-ml}-q^{-nk})e^{(n+l)M}(q^{(-m-k)L}-q^{(-m+k)L}).
\end{eqnarray*}
Now it is time to identify the algebraic structure of the
additional $\partial_{t_{l,k}^*}$ flows of the dispersionless BKP hierarchy in the following theorem.
\begin{theorem}
The additional flows $\partial_{t_{l,k}^*}$ of the dispersionless BKP hierarchy form a perturbative quantum torus algebra, i.e.,
\begin{equation}
[\partial_{t^*_{n,m}},\partial_{t^*_{l,k}}]L=[(q^{ml}-q^{nk})\partial_{t^*_{n+l,m+k}}-(q^{ml}-q^{-nk})\partial_{t^*_{n+l,m-k}}]L-\{A_{nmlk},L\},
\end{equation}
where
\[A_{nmlk}=(q^{ml}-q^{nk}+q^{-ml}-q^{-nk})e^{(n+l)M}(q^{(-m-k)L}-q^{(-m+k)L}).\]
\end{theorem}
\begin{proof}
Using the Jacobi identity, we can derive the following computation which will finish the proof of this theorem

\begin{eqnarray*}
&&[\partial_{t^*_{n,m,\beta}},\partial_{t^*_{l,k}}]L\\
&=&\partial_{t^*_{n,m}}(\{-(e^{lM}(q^{kL}-q^{-kL}))_-,L\})-\partial_{t^*_{l,k}}(\{-(e^{nM}(q^{mL}-q^{-mL}))_-,L\}) \\
&=&\{-(\partial_{t^*_{n,m}} (e^{lM}(q^{kL}-q^{-kL})))_-,L\} +\{-(e^{lM}(q^{kL}-q^{-kL}))_-,(\partial_{t^*_{n,m}} L)\}
\\
&&+\{\{-(e^{lM}(q^{kL}-q^{-kL}))_-,e^{nM}(q^{mL}-q^{-mL})\}_-,L\}\\
&&+\{(e^{nM}(q^{mL}-q^{-mL}))_-,\{-(e^{lM}(q^{kL}-q^{-kL}))_-,L\}\} \\
&=&\{\{(e^{nM}(q^{mL}-q^{-mL}))_-,e^{lM}(q^{kL}-q^{-kL})\}_-,L\}\\
&& +\{(e^{lM}(q^{kL}-q^{-kL}))_-,\{(e^{nM}(q^{mL}-q^{-mL}))_-,L\}\}
\\
&&+\{\{-(e^{lM}(q^{kL}-q^{-kL}))_-,e^{nM}(q^{mL}-q^{-mL})\}_-,L\}\\
&&+\{(e^{nM}(q^{mL}-q^{-mL}))_-,\{-(e^{lM}(q^{kL}-q^{-kL}))_-,L\}\} \\
&=&\{\{(e^{nM}(q^{mL}-q^{-mL}))_-,e^{lM}(q^{kL}-q^{-kL})\}_-,L\} \\&&+\{\{(e^{lM}(q^{kL}-q^{-kL}))_-,(e^{nM}(q^{mL}-q^{-mL}))_-\},L\}
\\
&&+\{\{-(e^{lM}(q^{kL}-q^{-kL}))_-,e^{nM}(q^{mL}-q^{-mL})\}_-,L\}\\
&=&\{\{e^{nM}(q^{mL}-q^{-mL}),e^{lM}(q^{kL}-q^{-kL})\}_-,L\}\\
&=&-\{\{[(q^{ml}-q^{nk})e^{(n+l)M}(q^{(m+k)L}-q^{-(m+k)L})-(q^{ml}-q^{-nk})e^{(n+l)M}(q^{(m-k)L}-q^{(-m+k)L})\\
&&+(q^{ml}-q^{nk}+q^{-ml}-q^{-nk})e^{(n+l)M}(q^{(-m-k)L}-q^{(-m+k)L})]\}_-,L\}\\
&=&[(q^{ml}-q^{nk})\partial_{t^*_{n+l,m+k}}-(q^{ml}-q^{-nk})\partial_{t^*_{n+l,m-k}}]L-\{(A_{nmlk})_-,L\}.
\end{eqnarray*}

\end{proof}
Comparing the above proposition with \cite{torus} in which the complete quantum torus symmetry of the BKP hierarchy was constructed, the following remark should be noted.
\begin{remark}
The classical quantum torus structure
\begin{equation}
[\partial_{t^*_{n,m}},\partial_{t^*_{l,k}}]=(q^{ml}-q^{nk})\partial_{t^*_{n+l,m+k}}-(q^{ml}-q^{-nk})\partial_{t^*_{n+l,m-k}},
\end{equation}
is broken from the BKP hierarchy to its dispersionless hierarchy of which there exists one more perturbative term $\{A_{nmlk},L\}$.
\end{remark}

\section{ Multicomponent BKP hierarchy}

For an $N$-component BKP hierarchy, there are $N$ infinite families of time variables $t_{\alpha, n}, \alpha=1,\ldots,N, n=1,3,5,7,\ldots$. The coefficients $A,u_1, u_2,\ldots$ of the Lax operator
\begin{equation}\label{Lax}\L_B=A\pa +u_1\pa^{-1}+u_2\pa^{-2}+\ldots
\end{equation}
 are $N\times N$ matrices and $A=diag(a_1,a_2,\dots,a_N)$. There are another $N$ pseudo-differential operators $R_1,\ldots, R_N$ in the form of
$$R_{\alpha}=E_{\alpha}+u_{\alpha, 1}\pa^{-1}+u_{\alpha, 2}\pa^{-2}+\ldots,$$ where $E_{\alpha}$ is the $N\times N$ matrix with ``$1$" on the $(\alpha,\alpha)$-component and zero elsewhere, and $u_{\alpha, 1}, u_{\alpha, 2},\ldots$ are also $N\times N$ matrices. The operators $\L_B, R_1,\ldots, R_N$ satisfy the following conditions:
$$\L_BR_{\alpha}=R_{\alpha} \L_B, \quad R_{\alpha}R_{\beta}=\delta_{\alpha\beta}R_{\alpha}, \quad \sum_{\alpha=1}^N R_{\alpha}= E.$$
 Let $*$ denote a formal adjoint operation defined by $p^{*}=\sum(-1)^{i}\partial^{i}p_i^T$ for an arbitrary
matrix-valued pseudo-differential operator $p=\sum p_i\partial^{i}$ , and $(AB)^{*}=B^{*}A^{*}$ for two matrix-valued pseudo-differential operators $A,B$.
Here $\L_B$ must satisfy
\begin{equation}\label{abcd}
\L_B^*=-\partial \L_B \partial^{-1}.
\end{equation}

The Lax equations of the multicomponent BKP hierarchy are:
\begin{equation*}
\frac{\pa \L_B}{\pa t_{n}^{\alpha}}=[B_{\alpha, n}, \L_B],\hspace{1cm}\frac{\pa R_{\beta}}{\pa t_{n}^{\alpha}}=[B_{\alpha, n}, R_{\beta}],\hspace{1cm}B_{\alpha, n}:=(\L_B^n R_{\alpha})_+,\  \ n\in\Zop.
\end{equation*}
 The operator $\pa$ now is equal to $a_1^{-1}\pa_{t_{1}^1}+\ldots +a_N^{-1}\pa_{t_{1}^N}$.
In fact the Lax operators $\L_B$ and $R_{\alpha}$ can have the following dressing structures
\[\L_B=\Phi A\d \Phi^{-1}, \ \ R_{\alpha}=\Phi E_{\alpha}\Phi^{-1}.\]

Then the dressing operator $\Phi$ needs to satisfy \begin{equation}\label{phipsi}
\Phi^*=\d \Phi^{-1}\d^{-1}.
\end{equation}

We call the eq.\eqref{abcd} the B type condition of the $N$-component BKP hierarchy.

We can get the operators $B_{\alpha, n},R_j$ satisfy the following B type condition
\begin{equation}\label{phipsi}
B_{\alpha, n}^*=-\d B_{\alpha, n}\d^{-1}, R^{\ast}_j=\d R_j\d^{-1},
\end{equation}

and the dressing operator $\Phi$ satisfies the following Sato equations
\begin{equation*}
\frac{\pa \Phi}{\pa t_{n}^{\alpha}}=-(\L_B^n R_{\alpha})_-\Phi.
\end{equation*}

\section{Quantum torus symmetry of the multicomponent BKP hierarchy}

To construct the additional quantum torus symmetry of the multicomponent BKP hierarchy, firstly we define the operator $\Gamma_B$ and the Orlov-Shulman's  operator $\M_B$  as
 \begin{equation}
\Gamma_B=\sum_{i\in \Zop } \sum_{j=1}^Nit_i^jA^{-1}E_{jj}\partial^{i-1},\ \ \M_B= \Phi \Gamma_B \Phi^{-1}.
\end{equation}
The Lax operator $\L_B$ and the Orlov-Shulman's operator $\M_B$  satisfy the following matrix canonical relation
\[[\L_B,\M_B]=E.\]

Given an operator $\L_B$, the dressing operators $\Phi$ are determined uniquely up to a multiplication to the
right by operators with
constant coefficients.

We denote $t=(t_1,t_3,t_5,\dots)$ and introduce
the wave function as
\begin{align}\label{wavef}
w_B(t; z)=\Phi e^{\xi_B(t;z)},
\end{align}
where the matrix function $\xi_B$ is defined as $\xi_B(t;
z)=\sum_{k\in\Zop} \sum_{j=1}^Nt_k^jE_{jj} z^k$. It is easy to see
$\d^i e^{x z}=z^i e^{x z},\ \ i\in\Z$
and
\[
\L_B\,w_B(t;z)=z w_B(t;z),\ \ \frac{\d w_B}{\d t_{2n+1}^j}=(\L^{2n+1}_{B}R_{j})_+w_B.
\]

With the above preparation, it is time to  construct additional symmetries for the  multicomponent BKP hierarchy in the next part.
It is easy to get  that the operator $\M_B$  satisfy

\begin{equation}
[\L_B, \M_B]=1,  \
\M_B w_B(z)=\d_z w_B(z);
\end{equation}
\begin{equation}\label{bkpMt}
\frac{\d \M_B}{\d t_k^j}=[(\L_B^kR_{j})_+,\M_B],\ \ k\in\Zop.
\end{equation}

Given any pair of integers $(m,n)$ with $m,n\ge0$, we will introduce the following  matrix-valued  operator $B_{m nj}$
\begin{align}\label{defBoperator}
B_{m nj}=\M_B^m\L_B^{n}R_{j}-(-1)^{n} R_{j}\L_B^{n-1}\M_B^{m}\L_B.
\end{align}

For any  matrix operator $B_{m nj}$ in \eqref{defBoperator}, one has
\begin{align}\label{Bflow}
&\frac{\d B_{m nj}}{\d t_k^j}=[(\L_B^kR_{j})_+, B_{m nj}],  \ k\in\Zop.
\end{align}

To prove that the $B_{m nj}$  satisfies the B type condition, we need the following lemma.
\begin{lemma}\label{BtypM}
The  matrix operator $\M_B$ satisfies the following identity,
\[\label{MBproperty}
\M_B^*
=\d\L_B^{-1}\M_B\L _B\d^{-1}.\]
\end{lemma}
\begin{proof}
Using identities as
\[
 \Phi^*=\d\Phi^{-1} \d^{-1},\ \ \Gamma_B^*=\Gamma_B;\]
the  following calculations
\[
\M_B^* =\Phi^{*-1}\Gamma_B \Phi^*=\d\Phi \d^{-1}\Gamma_B \d\Phi^{-1} \d^{-1}
=\d\Phi \d^{-1}\Phi^{-1}\M_B\Phi \d\Phi^{-1} \d^{-1},\]
will lead to the lemma.
\end{proof}
 Basing on the Lemma \ref{BtypM} above, it is easy to check that the  matrix-valued operator
$B_{m nj}$  satisfy the B type condition, namely
\begin{equation}\label{btypeB}
B_{m nj}^*=-\d  B_{m nj} \d^{-1}.
\end{equation}

Now we will denote the  matrix operator $D_{m nj}$ as
\begin{equation}
D_{m nj}:=e^{m\M_B}q^{n\L_B}R_{j}-\L_B^{-1}R_{j}q^{-n\L_B}e^{m\M_B}\L_B,
\end{equation}
which further leads to
\begin{equation}
D_{m nj}=\sum_{p,s=0}^{\infty}\frac{m^p(n\log q)^s(\M_B^p\L_B^sR_{j}-(-1)^sR_{j}\L_B^{s-1}\M_B^p\L_B)}{p!s!}=\sum_{p,s=0}^{\infty}\frac{m^p(n\log q)^sB_{p sj}}{p!s!}.
\end{equation}
Using the eq. \eqref{btypeB}, the following calculation will lead to the B type anti-symmetry property of $D_{m nj}$ as
\begin{eqnarray*}D_{m nj}^*
&=&(\sum_{p,s=0}^{\infty}\frac{m^p(n\log q)^sB_{p sj}}{p!s!})^*\\
&=&-(\sum_{p,s=0}^{\infty}\frac{m^p(n\log q)^s\d B_{p sj}\d^{-1}}{p!s!})\\
&=&-\d(\sum_{p,s=0}^{\infty}\frac{m^p(n\log q)^sB_{p sj}}{p!s!})\d^{-1}\\
&=&-\d  D_{m nj} \d^{-1}.
\end{eqnarray*}
Therefore we get the following important B type condition, i.e.  the  matrix operator $D_{m nj}$ satisfies
\begin{equation}
D_{m nj}^*=-\d  D_{m nj} \d^{-1}.
\end{equation}

Then basing on a quantum parameter $q$, the additional flows for the time variable $t_{m,n}^j,t_{m,n}^{*j}$ are
defined as follows
\begin{equation}
\dfrac{\partial \Phi}{\partial t_{m,n}^j}=-(B_{m nj})_-\Phi,\ \
 \dfrac{\partial \Phi}{\partial t^{*j}_{m,n}}=-(D_{m nj})_-\Phi,
\end{equation}
or equivalently rewritten as

\begin{equation}
\dfrac{\partial \L_B}{\partial t_{m,n}^j}=-[(B_{m nj})_-,\L_B], \qquad
\dfrac{\partial \M_B}{\partial t^{*j}_{m,n}}=-[(D_{m nj})_-,\M_B].
\end{equation}

 Generally, one can also derive
\begin{equation}\label{bkpMLK}
\partial_{t^{*i}_{l,k}}(D_{m nj})=[-(D_{l ki})_-,D_{m nj}].
\end{equation}

The way in the construction of the  matrix-valued operator $D_{m nj}$ depends on the reduction condition in  the eq. (\ref{abcd})  on the generators of the additional flows.

 This further leads to the commutativity of the additional flow $\dfrac{\partial }{\partial
t^{*j}_{m,n}}$  with the flow $\dfrac{\partial
}{\partial t_k^j}$ in the following theorem.

\begin{theorem}
The additional flows of $\partial_{t^{*s}_{l,k}}$ are  symmetries of the  multicomponent BKP hierarchy, i.e. they commute with all $\partial_{t_n^j}$ flows of the   multicomponent BKP hierarchy.
\end{theorem}
\begin{proof}
The proof is similar as the KP hierarchy by using the Theorem 6.2 in \cite{torus}, i.e.
the additional flows of $\partial_{t_{l,k}^{*s}}$ can commute with all $\partial_{t_n^j}$ flows of the   multicomponent BKP hierarchy. The detail will be omitted here.
\end{proof}
Comparing with the additional symmetry of the single-component BKP hierarchy in \cite{TU}, the additional flows $\partial_{t_{l,k}^s}$ of the  multicomponent BKP hierarchy form the following $N$-folds direct product of the
$W_{\infty}$ algebra  as following
\begin{eqnarray*}
&&[\partial_{t_{p,s}^r},\partial_{t_{a,b}^c}]\L_B=\delta_{rc}\sum_{\alpha\beta}C_{\alpha\beta}^{(ps)(ab)}\partial_{t_{\alpha,\beta}^c}\L_B.
\end{eqnarray*}

Now it is time to identity the algebraic structure of the
additional $t_{l,k}^{*j}$ flows of the  multicomponent BKP hierarchy.
\begin{theorem}\label{bkpalg}
The additional flows $\partial_{t^{*i}_{l,k}}$ of the  multicomponent BKP hierarchy form the $\bigotimes^NQT_+ $ algebra
(an $N$-folds direct product of the
positive half of the quantum torus algebra $QT$), i.e.,
\begin{equation}
[\partial_{t^{*r}_{n,m}},\partial_{t^{*j}_{l,k}}]=\delta_{rs}(q^{ml}-q^{nk})\partial_{t^{*r}_{n+l,m+k}},\ \ n,m,l,k\geq 0, \ \ 1\leq r,j\leq N.
\end{equation}

\end{theorem}

\begin{proof}

One can also prove this theorem as following by rewriting the quantum torus flow in terms of a combination of $\partial_{t_{m,n}^j}$ flows
\begin{eqnarray*}
&&[\partial_{t^{*r}_{n,m}},\partial_{t^{*j}_{l,k}}]\L_B\\
&=&[\sum_{p,s=0}^{\infty}\frac{n^p(m\log q)^s}{p!s!}\partial_{t_{p,s}^r},\sum_{a,b=0}^{\infty}\frac{l^a(k\log q)^b}{a!b!}\partial_{t_{a,b}^j}]\L_B\\
&=&\sum_{p,s=0}^{\infty}\sum_{a,b=0}^{\infty}\frac{n^p(m\log q)^s}{p!s!}\frac{l^a(k\log q)^b}{a!b!}[\partial_{t_{p,s}^r},\partial_{t_{a,b}^j}]\L_B\\
&=&\sum_{p,s=0}^{\infty}\sum_{a,b=0}^{\infty}\frac{n^p(m\log q)^s}{p!s!}\frac{l^a(k\log q)^b}{a!b!}\sum_{\alpha\beta}C_{\alpha\beta}^{(ps)(ab)}\delta_{rj}\partial_{t_{\alpha,\beta}^r}\L_B\\
&=&(q^{ml}-q^{nk})\sum_{\alpha,\beta=0}^{\infty}\frac{(n+l)^\alpha((m+k)\log q)^\beta}{\alpha!\beta!}\delta_{rj}\partial_{t_{\alpha,\beta}^r}\L_B\\
&=&(q^{ml}-q^{nk})\delta_{rj}\partial_{t^{*r}_{n+l,m+k}}\L_B.
\end{eqnarray*}
\end{proof}

From these above, we can find that the quantum torus symmetry can be generalized from the BKP hierarchy to the multicomponent BKP hierarchy by the multi-folds product.

{\bf {Acknowledgements:}}
 This work is supported by the National Natural Science Foundation of China under Grant No. 11571192 and the K. C. Wong Magna Fund in
Ningbo University.

\end{document}